\documentclass{ecai}
\usepackage{times}
\usepackage{graphicx}
\usepackage{latexsym}

\usepackage{wrapfig}

\usepackage{mathtools}
\usepackage{graphicx}
\usepackage{courier}
\usepackage{epstopdf}
\usepackage{color}
\usepackage{csquotes}

\usepackage{amsmath}
\usepackage{amsfonts}
\usepackage{amssymb}
\usepackage[subnum]{cases}
\usepackage{booktabs} 
\usepackage{pifont}
\usepackage{graphics}

\usepackage{amssymb}
\usepackage{pifont}

\graphicspath{ {./Images/} }

\usepackage {tikz}
\usetikzlibrary {positioning}
\definecolor {processblue}{cmyk}{0.96,0,0,0}
\usepackage{lipsum,adjustbox}

\usepackage[subnum]{cases}
\usepackage{color,soul}

\usepackage{algorithm,algorithmicx,algpseudocode}

\usepackage{subcaption}

\usepackage{relsize}

\newtheorem{lemma}{Lemma}
\newtheorem{proof}{Proof}

\usepackage{footnote}
\makesavenoteenv{tabular}

\usepackage{multirow}

\usepackage{hyperref}
\usepackage{cleveref}[2012/02/15]

\crefformat{footnote}{#2\footnotemark[#1]#3}

\begin{document}

\title{Integrating Network Embedding and Community Outlier Detection via Multiclass Graph Description}

\author{Sambaran Bandyopadhyay\institute{IBM Research \& IISc, Bangalore, email: samb.bandyo@gmail.com} \and Saley Vishal Vivek\institute{Indian Institute of Science, Bangalore, email: vishalsaley@iisc.ac.in} \and M. N. Murty\institute{Indian Institute of Science, Bangalore, email: mnm@iisc.ac.in} }

\maketitle
\bibliographystyle{ecai}

\begin{abstract}
Network (or graph) embedding is the task to map the nodes of a graph to a lower dimensional vector space, such that it preserves the graph properties and facilitates the downstream network mining tasks. Real world networks often come with (community) outlier nodes, which behave differently from the regular nodes of the community. These outlier nodes can affect the embedding of the regular nodes, if not handled carefully. In this paper, we propose a novel unsupervised graph embedding approach (called DMGD) which integrates outlier and community detection with node embedding. We extend the idea of deep support vector data description to the framework of graph embedding when there are multiple communities present in the given network, and an outlier is characterized relative to its community. We also show the theoretical bounds on the number of outliers detected by DMGD. Our formulation boils down to an interesting minimax game between the outliers, community assignments and the node embedding function. We also propose an efficient algorithm to solve this optimization framework. Experimental results on both synthetic and real world networks show the merit of our approach compared to state-of-the-arts.
\end{abstract}

\begin{figure*}
\begin{subfigure}[t]{.24\textwidth}
  \centering
  \includegraphics[width=\linewidth]{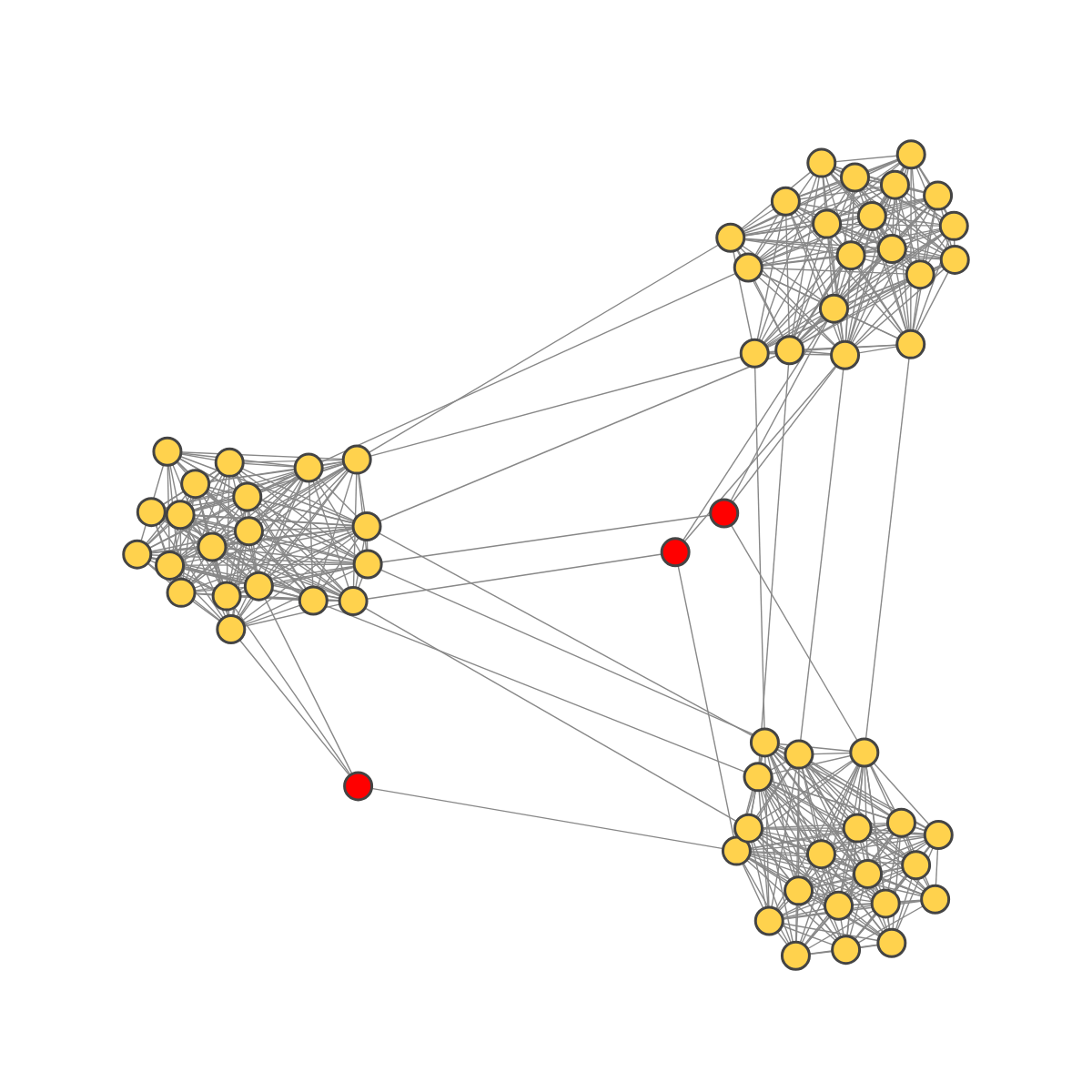}
  \caption{Input Graph}
\end{subfigure}
\begin{subfigure}[t]{.22\textwidth}
  \centering
  \includegraphics[width=\linewidth]{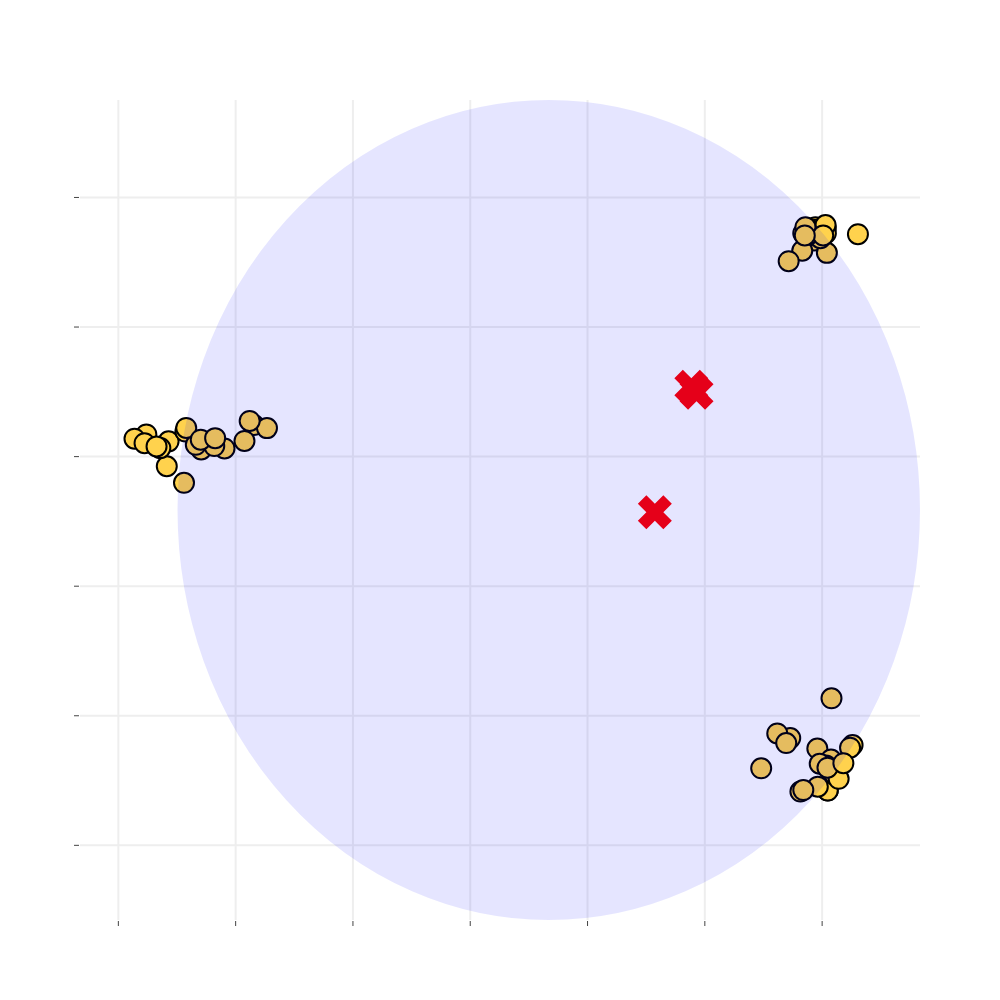}
  \caption{SVDD}
\end{subfigure}
\begin{subfigure}[t]{.28\textwidth}
  \centering
  \includegraphics[width=\linewidth]{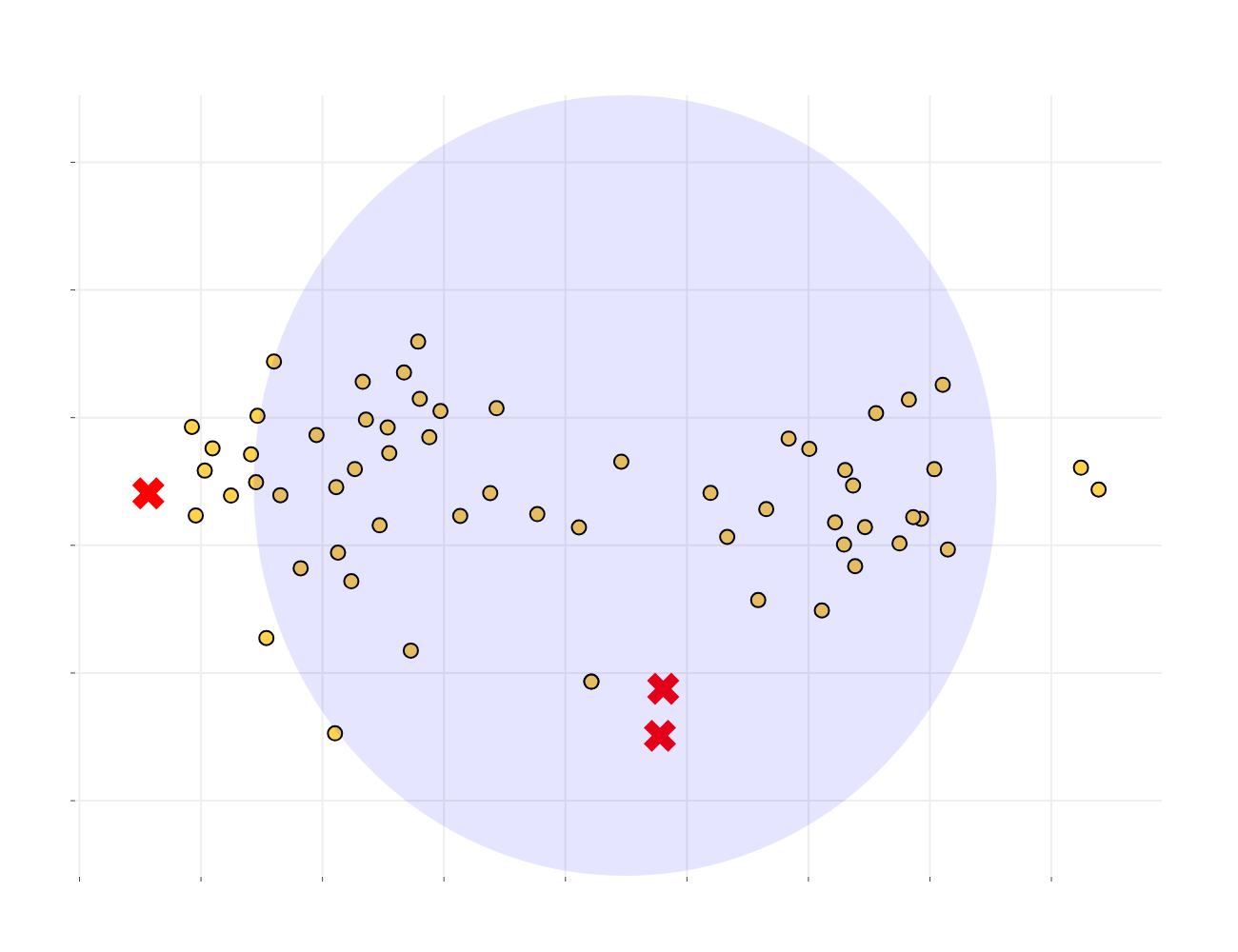}
  \caption{Deep SVDD}
\end{subfigure}
\begin{subfigure}[t]{.22\textwidth}
  \centering
  \includegraphics[width=\linewidth]{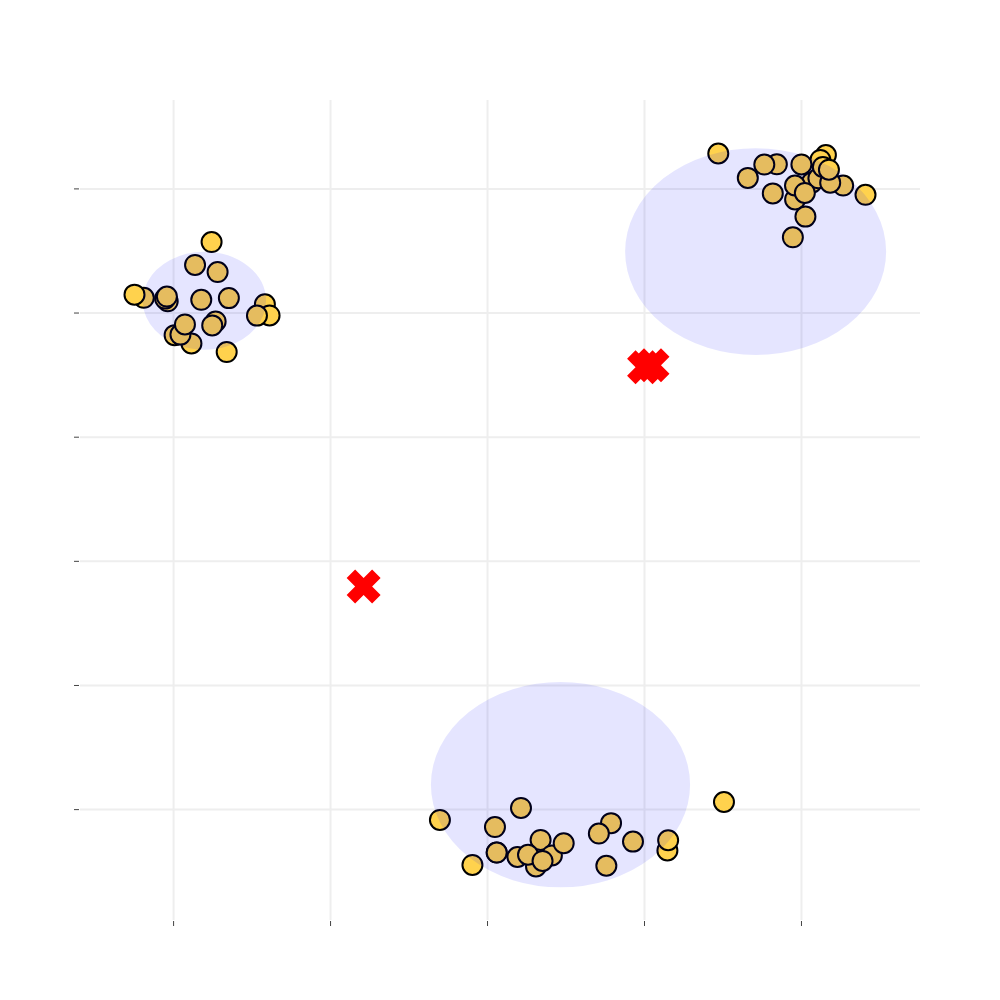}
  \caption{DMGD}
\end{subfigure}
\caption{We motivate and compare our proposed graph embedding algorithm DMGD with SVDD and Deep SVDD on a small synthetic network with 150 regular nodes divided into 3 communities. There are also 3 community outliers (marked in red) as they do not adhere the community structure of the network. We use Eq. \ref{eq:loss1} to generate the node embeddings, and then feed them to SVDD and Deep SVDD. We keep the embedding dimension as 2 to plot the node embeddings. As expected, SVDD and Deep SVDD do not respect the community structure of the network and hence mostly include all the outliers within the spheres they form. But DMGD (proposed algo.) detects all the outliers by keeping them outside of the learned community boundaries.}
\label{fig:motiv}
\end{figure*}

\section{\uppercase{Introduction}}\label{sec:intro}
Graphs are popularly used to model structured objects such as social and information networks. Given a graph $G=(V,E)$ with $N$ nodes, network embedding (also known as graph embedding or network representation learning) \cite{perozzi2014deepwalk,kipf2016semi} is the task to learn a function $f: V \rightarrow \mathbb{R}^M$, i.e., which maps each node of the graph to a vector of dimension $M < N$. 
The goal of graph embedding is to preserve the underlying graph structure in the embedding vector space. Typically, the quality of embedding is validated on several downstream graph mining tasks such as node classification, community detection (node clustering), etc. Different types of graph embedding techniques exist in the literature, such as random walk based embedding \cite{perozzi2014deepwalk,grover2016node2vec}, graph reconstruction based embedding \cite{wang2016structural,gao2018deep}, graph neural network based embedding \cite{kipf2016semi,hamilton2017inductive,velivckovic2018deep}, etc. Fundamentally, many of these algorithms work on the assumption of homophily \cite{mcpherson2001birds} property and the community structures that most of the networks exhibit. These properties ensure that nodes which are directly connected or closer to each other in the graph, tend to be similar to each other in attributes and form a community in the graph. 

Most of the above graph embedding algorithms perform good when the nodes behave as expected. But real world networks often contain nodes which are outliers in nature. These outlier nodes behave differently in terms of their connections to other nodes and attribute values, compared to most of the nodes in their respective communities (that's why they are often called community outliers). For example, an outlier node can almost be uniformly connected to nodes from different communities, thus violating the community structure of the network. In this work, we have used the phrases \textit{outlier} and \textit{community outlier} interchangeably.
Detection of community outliers has been studied in \cite{gao2010community,ding2019deep}.
But these outlier nodes, though are smaller in number typically, can significantly affect the embedding of the normal nodes, if not treated specially while generating the embeddings. It has been observed that mere post processing of the embeddings cannot filter out the outliers as the embeddings are already affected by them \cite{bandyopadhyay2018outlier}. Recently, \cite{liang2018semi} proposed a semi-supervised algorithm based on reconstruction loss, combined with node classification error of an autoencoder and \cite{bandyopadhyay2018outlier} proposed an unsupervised matrix factorization based approach to deal with outliers in network embedding, but strictly for attributed networks.

In general, the concept of outlier detection (a.k.a. one class classification) is a widely studied problem in machine learning \cite{scholkopf2001estimating,tax2004support,le2011multiple,wu2009small}. In this paper, we would focus on a recent approach, called deep support vector data description (deep SVDD) \cite{ruff2018deep}. Deep SVDD is a deep learning based extension of SVDD \cite{tax2004support}. In SVDD, regular data points are mapped to be within a sphere, while outliers stay outside. Deep SVDD uses deep neural networks to learn this mapping. Though Deep SVDD can be applied for detecting outliers in graphs, it is not suitable when the graph has multiple communities or clusters. As shown in Fig. \ref{fig:motiv}, an outlier between two communities can actually be marked as a non-outlier by an approach similar to Deep SVDD or SVDD, where they form only one sphere for the whole graph. Also this type of mapping is not suitable for representation learning on graphs, as there is no explicit way to preserve other  characteristics of the graph while detecting the outliers.

\textbf{Contributions:} 
We propose a novel deep learning based unsupervised algorithm (referred as DMGD - \textbf{D}eep \textbf{M}ulticlass \textbf{G}raph \textbf{D}escription) which extends the idea of support vector data description to jointly learn community outliers and node embedding by minimizing the effect of outliers in the embedding space. Our approach meets the requirements when the graph has multiple communities and an outlier is a node not being enclosed by any community. DMGD unifies node representation, outlier detection and community detection in graphs through a single optimization framework which boils down to an interesting minimax game. We have shown the \textbf{\textit{theoretical bounds}} on the number of outliers detected by DMGD. Experimental results depict the merit of DMGD on both synthetic and real life network datasets for various downstream network mining tasks.
Source code of DMGD can be found at \url{https://github.com/vasco95/DMGD} to ease the reproducibility of the results.

\section{\uppercase{Related Work}}\label{sec:graphEmb}
Detailed surveys on graph representation can be found in \cite{hamilton2017representation,wu2019comprehensive}. We briefly discuss some important graph embedding techniques here. The concept of representing words in a corpus by vectors \cite{mikolov2013distributed} in NLP literature influenced some early work in network representation learning. DeepWalk \cite{perozzi2014deepwalk}, node2vec \cite{grover2016node2vec} use random walk on the graph to capture nodes similar to a node and generate similar embeddings for the nodes which are close and frequently reachable from each other. struc2vec \cite{ribeiro2017struc2vec} is another random walk based technique where structurally similar nodes are assigned similar embeddings, even if they are far from each other in the graph. There are deep autoencoder based graph embedding techniques such as SDNE \cite{wang2016structural} and DNGR \cite{cao2016deep} which preserve different orders of proximities of the graph in the embedding space. TADW \cite{yang2015network}, AANE \cite{huang2017accelerated} and DANE \cite{gao2018deep} use complimentary information from the attributes associated with the nodes in the reconstruction of the graph properties (via matrix factorization and deep autoencoders) to generate node embeddings. Along with node proximity, global node ranking of the graph is preserved in the embedding space in \cite{lai2017prune}.

Graph neural networks \cite{scarselli2009graph} gained significant importance in the recent literature. A semi-supervised graph convolutional network (GCN) is proposed by recursively aggregating attribute information from the neighborhood of each node in \cite{kipf2016semi}. GraphSAGE \cite{hamilton2017inductive} is an inductive representation technique which proposes to aggregate different types of neighborhood aggregation functions in GCN. A scalable and faster version of GCN via neighborhood subsampling technique is proposed in \cite{chen2018fastgcn}. Attention mechanisms for graph embedding are proposed in GAT \cite{velivckovic2018deep} and in Graph Attention \cite{abu2018watch}.

None of the above techniques explicitly minimize the effect of outlier nodes in graph embeddings. However, real life social networks often come with outlier nodes, which can affect the embedding of the other nodes of the graph. Recently, \cite{liang2018semi} proposes a semi-supervised approach, SEANO, which learns outliers in network embedding framework. An unsupervised approach, ONE, is proposed in \cite{bandyopadhyay2018outlier} to minimize the effect of outliers by weighted matrix factorization for attributed network embedding.
Extending the idea of ONE, two deep neural architectures are proposed \cite{bandyopadhyay2020outlier} to minimize the effect of outliers on the node embeddings, again for the attributed networks.
These three approaches are based on the attributes present in the nodes and exploit the inconsistency between link structures and node attributes in the graph to detect the outliers.
In this paper, we propose an integrated unsupervised approach by extending the idea of SVDD from one cluster to multiple clusters, and pose it as a representation learning problem for outlier and community detection in graphs. Contrary to the existing literature, our proposed algorithm DMGD can work by focusing only on the network link structure to detect and minimize the effect of outliers in node embeddings.

\section{\uppercase{Preliminaries of SVDD and Deep SVDD}}\label{sec:svdd}
SVDD \cite{tax2004support}, inspired by support vector classifier, obtains a spherically shaped boundary around the regular points of the given dataset, characterizing outliers as the points which stay outside of the sphere. More formally, for a set of points $x_i \in \mathbb{R}^N$, SVDD aims to find the smallest hypersphere with center at $c \in \mathcal{F}$ (in feature space) and radius $R > 0$ which encloses most of the points, as follows.
\begin{equation}\label{eq:SVDD}
\begin{aligned}
& \underset{R,c,\mathbf{\xi}}{\text{min}}
& & R^2 + \alpha \sum\limits_{i} \xi_i \\
& \text{such that,}
& & ||\phi(x_i) - c||_\mathcal{F} \leq R^2 + \xi_i, \;\; \xi_i \geq 0 \;\; \forall i
\end{aligned}
\end{equation}
$\phi(x_i)$ is a function that maps the data points to a feature space $\mathcal{F}$. $\alpha > 0$ is a weight parameter and $\xi_i$ are the slack variables. Points for which $\xi_i >0$ stay outside of the sphere and are considered as outliers. 
Recently, Deep SVDD \cite{ruff2018deep} extends SVDD via deep learning. The (soft-boundary) Deep SVDD objective function is shown below:
\begin{equation}\label{eq:DeepSVDD}
\begin{aligned}
& \underset{R,\mathcal{W}}{\text{min}}
& R^2 + \alpha \sum\limits_{i} max\{0,||\phi(x_i;\mathcal{W})-c||^2 - R^2\} \\
&& + \frac{\lambda}{2} \sum\limits_{l=1}^L ||W^l||_F^2
\end{aligned}
\end{equation}
Deep SVDD replaces the function $\phi$ of SVDD with a deep neural network, with the set of parameters $\mathcal{W}$ which includes $L$ layers. Please note the second term of Eq. \ref{eq:DeepSVDD} is equivalent to having the slack variables in Eq. \ref{eq:SVDD}. To avoid trivial solution of the optimization problem, authors of \cite{ruff2018deep} do not include the center $c$ as an optimization variable. Rather, they fix it using some preprocessing.
Deep SVDD and SVDD suffer from a serious problem. They cannot distinguish outliers from the data when the given dataset has community structure and outliers can reside between communities, but not on the outskirt of the whole dataset, as shown in Fig. \ref{fig:motiv}. Besides as mentioned in Section \ref{sec:intro}, they are also not suitable for representation learning, as they do not ensure other important properties of the data objects to be preserved in the embedding (feature) space.

\section{\uppercase{Our Approach}}\label{sec:approach}
Here we discuss the proposed algorithm DMGD, which integrates graph embedding with outlier and community detection, in an unsupervised way. Given the graph $G=(V,E)$ with $|V|=N$,
DMGD learns a map $f: V \rightarrow \mathbb{R}^M$, where $M < N,D$. We also assume that there are $K$ unknown communities present in the graph. Our goal is to map $N$ vertices of the graph to $K$ communities, such that each regular point stays close to at least one of the communities, whereas, outliers stay outside of those communities. Along with that, we also want to preserve other graph properties in the embedding space, so that it facilitates the downstream graph mining tasks. For notational convenience, we define: $[N]= \{1,2,\cdots,N\}$, and similarly, $[K]=\{1,2,\cdots,K\}$.

First, we use a deep autoencoder to generate the initial graph embeddings. For each node $v_i$, the encoder function $f(a_i; \mathcal{W})$ maps the input structure vector to an $M$ dimensional space.
We use rows $a_i \in \mathbb{R}^N$, $\forall i \in [N]$ of the adjacency matrix of the graph $G$ as the structural vector. One can even use page rank vectors \cite{cao2016deep} to capture higher order proximities of the nodes or additional attributes (if available) to replace the structural vector.
There is also a decoder function $g(f(a_i), \mathcal{W})$ which maps the embedding of the node back to $\mathbb{R}^N$ space to reconstruct the input \footnote{We omit the parameter set $\mathcal{W}$ from the function definitions of $f$ and $g$ when there is no ambiguity}. $\mathcal{W}=\{W^1,\cdots,W^L\}$ contains parameters for $L$ layers of the autoencoder. We assume both encoder and decoder contain equal number of hidden layers. The autoencoder minimizes the reconstruction loss defined as: $\sum\limits_{i=1}^N || a_i - g(f(a_i)) ||_2^2$ with respect to the parameters $\mathcal{W}$ of the neural network. We also use the homophily property \cite{mcpherson2001birds} of an information network, which ensures two nodes which are directly connected by an edge to behave similarly. So, we minimize the L2 distance of the two embeddings where the corresponding nodes are connected by an edge: $\sum\limits_{(i,j) \in E}||f(a_i) - f(a_j) ||_2^2$. Thus, the total loss minimized to preserve these two properties is:
\begin{equation}\label{eq:loss1}
\begin{aligned}
& \underset{\mathcal{W}}{\text{min}}
& & \sum\limits_{i=1}^N || a_i - g(f(a_i)) ||_2^2 + \sum\limits_{(i,j) \in E}||f(a_i) - f(a_j) ||_2^2
\end{aligned}
\end{equation}
It is important to note that this formulation is very generic and can be replaced easily with alternate unsupervised techniques that are based on graph convolution autoencoders \cite{kipf2016variational} or random walks \cite{grover2016node2vec}.

Next, we integrate outlier and community detection with the graph embedding objective. In the process, we would also reduce the effect of outliers on the embedding of other regular nodes. Given, there are $K$ unknown communities in the input graph, we seek to obtain the centers of these $K$ communities. Let, $C$ contains these centers as, $C = \{c_1,\cdots,c_K\} \subset \mathbb{R}^M$. 
For each community, we like to find the smallest hypersphere \cite{tax2004support,le2011multiple} which encloses majority of the embeddings from that community. 
Let, $R_k > 0$ be the radius of the $k^{th}$ community, with $\mathcal{R} = \{R_1,\cdots,R_K\}$. For any node $v_i$, its community is determined by the sphere which encloses it (anyone if there are multiple such spheres) or by the smallest distance of the periphery of the spheres when it is outside of all the spheres (i.e., outliers). So the \textbf{community index} for the node $v_i$ is $\underset{k}{\text{argmin}} \; \text{max} \{ ||f(a_i) - c_k||_2^2 - R_k^2, \; 0 \}$. So, given the embeddings of the nodes as $f(a_i)$, $\forall i$, we optimize the following quantity:
\begin{equation}\label{eq:loss2}
\begin{aligned}
& \underset{\mathcal{R},\mathcal{C},\mathbf{\xi}}{\text{min}}
& & \sum\limits_{k=1}^K R_k^2 + \alpha \sum\limits_{i=1}^N \xi_i \\
& \text{such that}
& & \underset{k \in \{1,\cdots,K\}}{\text{min}} \; \{||f(a_i) - c_k||_2^2 - R_k^2\} \leq \xi_i \;\; \forall i \in [N] \\
&&& \xi_i \geq 0, \; \forall i \in [N] \;\; \text{and} \;\; R_k \geq 0, \; \forall k \in [K]
\end{aligned}
\end{equation}
Here $\xi_i \geq 0$ is the slack variable corresponding to the $i$th node of the graph. With respect to this formulation, nodes can be divided into three categories as follows. A \textbf{regular node} is one which stays strictly inside a community, and thus for it: $\underset{k}{\text{min}} \; \{||f(a_i) - c_k||_2^2 - R_k^2\} < 0$. A \textbf{boundary node} is one which lies exactly on the boundary of its community, so for it: $\underset{k}{\text{min}} \; \{||f(a_i) - c_k||_2^2 - R_k^2\} = 0$. An \textbf{outlier node} stays outside of all the communities, and thus for it: $\xi_i > 0$ (strictly positive).
These slack variables ensure soft spherical boundaries of the communities, and outlier nodes stay outside of the community. When the weight parameter $\alpha$ is very small, an optimizer would mainly focus to minimize the first term in the cost function in Eq. \ref{eq:loss2}, leading to very small (in terms of radius) spherical communities and many nodes will be treated as outliers. Whereas, a higher value of $\alpha$ ensures lesser number of outliers with larger communities. 
Our approach implicitly assumes that the communities in the graph are spherical in the embedding space, which is not a hard requirement. Because of the soft boundaries of the spheres and the characterization of the outliers, it can also handle communities which are not exactly spherical.
The nonlinear properties of the original network are captured by neural networks for mapping into the embedding space.

Eq. \ref{eq:loss1} ensures that generic graph properties are preserved in the embedding space, while Eq. \ref{eq:loss2} ensures that the community structure is maintained in the embedding space, while separating outliers from the other nodes. So the combined objective of DMGD is given below:
\begin{equation}\label{eq:jointLoss}
\begin{aligned}
& \underset{\mathcal{R},\mathcal{C},\mathbf{\xi},\mathcal{W}}{\text{min}}
& & \sum\limits_{k=1}^K R_k^2 + \alpha \sum\limits_{i=1}^N \xi_i + \beta \sum\limits_{i=1}^N || a_i - g(f(a_i)) ||_2^2
\\
&&& + \gamma \sum\limits_{(i,j) \in E}||f(a_i) - f(a_j) ||_2^2\\
& \text{such that,}
& & \underset{k \in \{1,\cdots,K\}}{\text{min}} \; \{||f(a_i) - c_k||_2^2 - R_k^2\} \leq \xi_i \;\; \forall i \in [N] \\
&&& \xi_i \geq 0, \; \forall i \in [N] \;\; \text{and} \;\; R_k \geq 0, \; \forall k \in [K]
\end{aligned}
\end{equation}
Following lemmas show the formal connection between the number of outliers and the parameter $\alpha$.
\begin{lemma}\label{lemma:upbound}
\textbf{Number of outlier nodes ($|\{v_i : \xi_i > 0 \}|$) detected by DMGD is upper bounded by $\frac{K}{\alpha}$.}
\end{lemma}
\begin{proof}
Suppose, the communities detected by DMGD are denoted by $C_k$, $k=1,\cdots,K$. The community index for node $v_i$ is $\underset{k}{\text{argmin}} \; \text{max} \{ ||f(a_i) - c_k||_2^2 - R_k^2, \; 0 \}$. We use $\nu$-property as stated in \cite{scholkopf2001estimating} to prove the claim. Partial objective function of DMGD can be written as:
\begin{align*}
    \sum\limits_{k=1}^K R_k^2 + \alpha \sum\limits_{i=1}^N \xi_i \; = \; \sum\limits_{k=1}^K \Big( R_k^2 + \frac{1}{\frac{1}{\alpha}} \sum\limits_{i \in C_k} \xi_i \Big)
\end{align*}
So, when the communities are fixed, the objective for each community is exactly same as SVDD and thus $\nu$-property ensures that the number of outliers from each community is upper bounded by $\frac{1}{\alpha}$. Hence total number of outlier nodes detected by DMGD is upper bounded by $\frac{K}{\alpha}$.
\end{proof}
\begin{lemma}\label{lemma:lowbound}
\textbf{The sum of boundary nodes (nodes which lie exactly on the boundaries of their respective communities) and outlier nodes detected by DMGD is lower bounded by $\frac{K}{\alpha}$.}
\end{lemma}
This can also be proved similarly using the $\nu$-property. Lemmas \ref{lemma:upbound} and \ref{lemma:lowbound} together show that the number of outlier nodes detected by DMGD is actually bounded tightly, assuming not many nodes lie exactly on the community boundaries.

\subsection{Optimization and Training}\label{sec:opt}
Optimizing Eq. \ref{eq:jointLoss} is difficult because of multiple reasons. One primary reason is the presence of the constraints of type $\underset{k}{\text{min}} \; \{||f(a_i) - c_k||_2^2 - R_k^2\} \leq \xi_i$. So we use the following trick to replace them with some other variables, as follows:
\begin{equation}\label{eq:altLoss}
\begin{aligned}
& \underset{\mathcal{R},\mathcal{C},\mathbf{\xi},\mathcal{W}, \Theta}{\text{min}}
& & P = \sum\limits_{k=1}^K R_k^2 + \alpha \sum\limits_{i=1}^N \xi_i + \beta \sum\limits_{i=1}^N || a_i - g(f(a_i)) ||_2^2 \\
&&& + \gamma \sum\limits_{(i,j) \in E}||f(a_i) - f(a_j) ||_2^2\\
& \text{such that,}
& & \sum\limits_{k=1}^K \theta_{ik}(||f(a_i) - c_k||_2^2 - R_k^2) \leq \xi_i \;\; \forall i \in [N], \\ 
&&& \sum\limits_{k=1}^K \theta_{ik} = 1, \;\;\; \theta_{ik} \geq 0, \;\; \forall i \in [N], \; \forall k \in [K],\\
&&& \xi_i \geq 0, \;\; \forall i \in [N] \;\; \text{and} \;\;  R_k \geq 0, \;\; \forall k \in [K]
\end{aligned}
\end{equation}
We assume, $\Theta =\{\theta_{ik} \; | \; \forall i, k\}$. Optimizing Eq. \ref{eq:altLoss} is simpler compared to Eq. \ref{eq:jointLoss}, as we replaced the minimum over some functions by a functions which are linear in $\Theta$.
\begin{lemma}\label{lemma:equiv}
\textbf{If $(\mathcal{R}^*,\mathcal{C}^*,\mathbf{\xi}^*,\mathcal{W}^*, \Theta^*)$ is a minimum of Eq. \ref{eq:altLoss}, then $(\mathcal{R}^*,\mathcal{C}^*,\mathbf{\xi}^*,\mathcal{W}^*)$ is a minimum of Eq. \ref{eq:jointLoss}.}
\end{lemma}
\begin{proof}
The feasible set of the optimization problem in Eq. \ref{eq:altLoss} is a super set of that in Eq. \ref{eq:jointLoss}. Hence the minimum value of Eq. \ref{eq:altLoss} is always less than or equal to the minimum value of Eq. \ref{eq:jointLoss}. Also, the loss function of both the optimization problems are the same. 
Let us denote the term: $r_{ik} = ||f(a_i) - c_k||_2^2 - R_k^2$. For any $i$, define $k^*(i)$ as the community index for which $r_{ik}$ is minimum, i.e., $k^*(i) = \underset{k}{\text{argmin}} \; r_{ik}$. We will write $k^*(i)$ as just $k^*$ when there is no ambiguity about $i$.

To prove the claim: 
First, let us prove, for any given point $(\mathcal{R},\mathcal{C},\mathbf{\xi},\mathcal{W}, \Theta)$, there exists a set of non-negative slack variables $\bar{\mathbf{\xi}}$ for which $P(\mathcal{R},\mathcal{C},\bar{\mathbf{\xi}},\mathcal{W}, \bar{\Theta}) \leq P(\mathcal{R},\mathcal{C},\mathbf{\xi},\mathcal{W}, \Theta)$, where $P$ is the cost function in Eq. \ref{eq:altLoss} and in $\bar{\Theta}$, for each $i$, $\theta_{ik^*}=1$, and $\theta_{ik}=0$, $\forall k \neq k^*$. 
Now, for any $i \in [N]$, $\sum\limits_{k \in [K]} \theta_{ik} r_{ik} \geq r_{ik^*}$, as $r_{ik^*} \leq r_{ik}$ and $\theta_{ik}\geq 0$, $\forall k$ with $\sum\limits_{k}\theta_{ik}=1$. For each $i$, set $\bar{\xi}_i = \xi_i - (\sum\limits_{k \in [K]} \theta_{ik} r_{ik} - r_{ik^*}) \geq 0$. Clearly, $\bar{\xi}_i \leq \xi_i$, and $\sum\limits_{i}\bar{\xi}_i \leq \sum\limits_{i}\xi_i$. Hence, $P(\mathcal{R},\mathcal{C},\bar{\mathbf{\xi}},\mathcal{W}, \bar{\Theta}) \leq P(\mathcal{R},\mathcal{C},\mathbf{\xi},\mathcal{W}, \Theta)$.

Thus if $(\mathcal{R}^*,\mathcal{C}^*,\mathbf{\xi}^*,\mathcal{W}^*, \Theta^*)$ is already a minimum point of Eq. \ref{eq:altLoss}, then $\bar{\mathbf{\xi}} = \mathbf{\xi}$. Hence, $(\mathcal{R}^*,\mathcal{C}^*,\mathbf{\xi}^*,\mathcal{W}^*, \bar{\Theta}^*)$ is a minimum of Eq. \ref{eq:altLoss} and also because of the definition of $\bar{\Theta}$, $(\mathcal{R},\mathcal{C},\bar{\mathbf{\xi}},\mathcal{W})$ belongs to the feasible set of Eq. \ref{eq:jointLoss} and is a minimum of the same.
\end{proof}
Lemma \ref{lemma:equiv} shows solving the optimization problem in Eq. \ref{eq:altLoss} is somewhat equivalent\footnote{There can other issues related to different local minima of the cost functions.} to solving the optimization problem in Eq. \ref{eq:jointLoss}. Hence, we will focus to solve Eq. \ref{eq:altLoss} now onward. Let us first derive the partial Lagrangian dual of the same with respect to the variables $\mathcal{R},\mathcal{C},\mathbf{\xi},\mathcal{W}, \Theta$, assuming $\mathcal{W}$ (parameters of the neural network) as constant. We will later update $\mathcal{W}$ by backpropagation.
\begin{align}
    \mathcal{L} = \sum\limits_{k=1}^K R_k^2 + \alpha \sum\limits_{i=1}^N \xi_i + & \sum\limits_{i=1}^N \lambda_i \sum\limits_{k=1}^K \Big[ \theta_{ik} \big( f(a_i) - c_k ||^2 - R_k^2 \big) \nonumber \\
    & - \xi_i \Big] - \sum\limits_{i=1}^N \eta_i \xi_i
\end{align}
Here, $\lambda_i$, $\eta_i$, $\forall i \in [N]$ are non negative Lagrangian constants. Equating the partial derivatives of $\mathcal{L}$ with respect to $R_k$, $c_k$, and $\xi_i$ to zero and using some algebraic manipulation, we get the following:
\begin{align}\label{eq:KKT}
    \sum\limits_{i=1}^N \lambda_I \theta_{ik} = 1, \; \forall k, \;\;\; & c_k = \sum\limits_{i=1}^N \lambda_i \theta_{ik} f(a_i), \; \forall k \in [K], \nonumber \\ 
    & 0 \leq \lambda_i \leq \alpha, \; \forall i \in [N]
\end{align}
Using the above constraints under KKT conditions, the (partial) Lagrangian dual can be written as:
\begin{align}\label{eq:parDual}
    \mathcal{D} = &\sum\limits_{i=1}^N \sum\limits_{k=1}^K \lambda_i \theta_{ik} f(a_i)^T f(a_i) \nonumber \\ 
    & - \sum\limits_{i=1}^N \sum\limits_{j=1}^N \sum\limits_{k=1}^K \lambda_i \lambda_j \theta_{ik} \theta_{jk} f(a_i)^T f(a_j)
\end{align}
Clearly we want to maximize the above w.r.t. $\Lambda$ (where, $\Lambda = (\lambda_1,\cdots,\lambda_N)^T \in \mathbb{R}^N$) and minimize with respect to $\Theta$. Hence, including the terms associated with the neural network parameters and the set of constraints, we seek to solve the following for DMGD.
\begin{equation}\label{eq:dual}
    \begin{aligned}
    &\underset{\Theta, \mathcal{W}}{min} \;\; \underset{\Lambda}{max} \;\;
    \sum\limits_{i=1}^N \lambda_i f(a_i)^T f(a_i) \\
    & - \sum\limits_{i=1}^N \sum\limits_{j=1}^N \sum\limits_{k=1}^K \lambda_i \lambda_j \theta_{ik} \theta_{jk} f(a_i)^T f(a_j) \\
    & + \beta \sum\limits_{i=1}^N || a_i - g(f(a_i)) ||_2^2 + \gamma \sum\limits_{(i,j) \in E}||f(a_i) - f(a_j) ||_2^2 \\
 & \text{such that,} \;\; \sum\limits_{i=1}^N \lambda_i \theta_{ik} = 1, \;\;\; 0 \leq \lambda_i \leq \alpha, \\ 
 & \sum\limits_{k=1}^K \theta_{ik} = 1, \;\;\; \theta_{ik} \geq 0, \;\;\; \forall i \in [N], \; \forall k \in [K]
\end{aligned}
\end{equation}
\textbf{Interpretation of the Optimization}: The above formulation is a very interesting minimax game, where the objective is to minimize the cost w.r.t. the community assignment variables $\Theta$ and neural network parameters $\mathcal{W}$; and to maximize the same w.r.t. the Lagrangian constants $\Lambda$. From KKT complementary slackness, it can be shown that, $0 \leq \lambda_i \leq \alpha$ for the boundary nodes and $\lambda_i = \alpha$ for the outlier nodes. For regular nodes, $\lambda_i=0$. Thus, only the third and fourth terms of the cost function in Eq. \ref{eq:dual} play role to generate the embedding of the regular points. But for a boundary or an outlier point, first two terms also contribute. If we focus on the second term $\sum\limits_{i=1}^N \sum\limits_{j=1}^N \sum\limits_{k=1}^K \lambda_i \lambda_j \theta_{ik} \theta_{jk} f(a_i)^T f(a_j)$ which needs to be minimized w.r.t. $\Lambda$, $\lambda_i$ would be high for a node which is less similar (small values of $f(a_i)^T f(a_j)$) to all other nodes in the community. This is another way of interpreting outliers in this framework. Whereas, the embedding function $f$, by maximizing the above, tries to keep the nodes in the same cluster close to each other. Minimizing the first term w.r.t. the neural network parameters is equivalent to having an L2 regularizer on the embeddings of the nodes. The nodes with very high L2 norms are more likely to be outliers because of this term. Rewriting the second term we get $\sum\limits_{i=1}^N \sum\limits_{j=1}^N \lambda_i \lambda_j f(a_i)^T f(a_j)(\sum\limits_{k=1}^K \theta_{ik} \theta_{jk})$. Maximization of this term wrt community assignments $(\theta_{i1},\cdots,\theta_{iK})$ will be such that the term $\sum\limits_{k=1}^K \theta_{ik} \theta_{jk}$ will be higher when product $\lambda_i \lambda_j f(a_i)^T f(a_j)$ is higher (node similar other outlier or boundary node). In other words, a node will be assigned a community which is most similar to it.

\begin{figure*}[h]
\begin{subfigure}[b]{\linewidth}
\centering \includegraphics[width=0.95\linewidth]{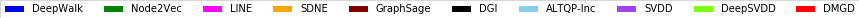}
\vspace{0.3cm}
\end{subfigure}
\begin{subfigure}[t]{.19\textwidth}
  \centering
  \includegraphics[width=\linewidth]{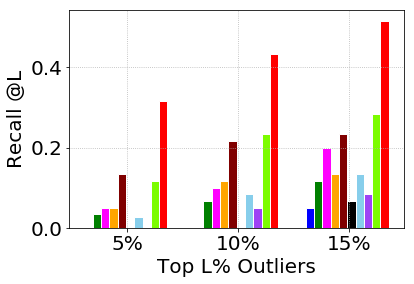}
  \caption{LFR Benchmark}
\end{subfigure}
\begin{subfigure}[t]{.19\textwidth}
  \centering
  \includegraphics[width=\linewidth]{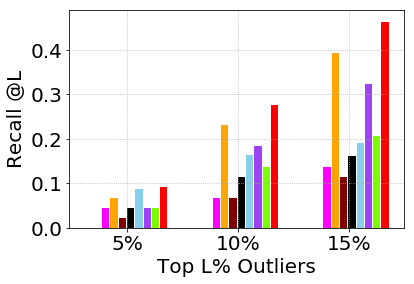}
  \caption{WebKB}
\end{subfigure}
\begin{subfigure}[t]{.19\textwidth}
  \centering
  \includegraphics[width=\linewidth]{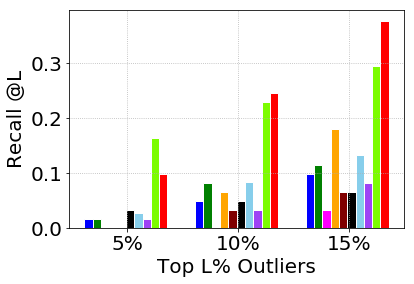}
  \caption{Polblogs}
\end{subfigure}
\begin{subfigure}[t]{.19\textwidth}
  \centering
  \includegraphics[width=\linewidth]{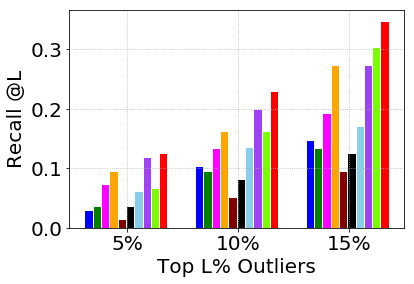}
  \caption{Cora}
\end{subfigure}
\begin{subfigure}[t]{.19\textwidth}
  \centering
  \includegraphics[width=\linewidth]{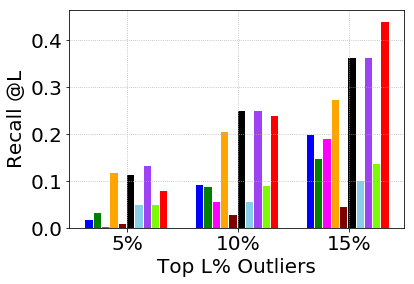}
  \caption{Pubmed}
\end{subfigure}
\caption{Recall at top L\% from the ranked list of outliers by different Embedding algorithms.}
\label{fig:outlierRecall}
\end{figure*}

\textbf{Parameter Initialization and Training of DMGD}: Like in any other deep learning technique, parameter initialization plays an important role in the convergence of our algorithm. First, we run few epochs of the autoencoder (optimizing just Eq. \ref{eq:loss1}) without considering the other terms. This gives generic embeddings $f(a_i)$ for each node $v_i \in V$. Then we run k-means++ algorithm \cite{kmeans++} (with number of clusters equal to $K$) on these embeddings to get the initial hard community assignments (i.e., values of $\theta_{ik}$). Lemma \ref{lemma:equiv} shows a minimum always corresponds to a hard community assignment. 

Once we have the initial embeddings and the community assignments, \textbf{(i)} we use standard quadratic solver CVXOPT qpsolver \footnote{https://cvxopt.org/examples/tutorial/qp.html} to update the values of $\Lambda$ (Eq. \ref{eq:dual} is a constrained quadratic w.r.t. $\Lambda$ when other variables are fixed).
\textbf{(ii)} Then we calculate the center of each community $k \in [K]$ by $c_k = \sum\limits_{i=1}^N \lambda_i \theta_{ik} f(a_i)$ from Eq. \ref{eq:KKT}.
\textbf{(iii)} For each community $k \in [K]$, we calculate the radius $R_k$ as $\underset{i}{\text{min}}\{||f(a_i)-c_k||_2^2 \; | \; \theta_{ik}=1 \; , \; 0 < \lambda_i < \alpha\}$. Experimentally we consider all the nodes for which $\lambda_i > 0$ and $\lambda_i < \alpha - 10^{-6}$ to compensate for numerical errors. \textbf{(iv)} We reassign the community of a node $v_i$, $\forall i \in [N]$ with the updated embeddings and centers by updating $\theta_{ik} = \underset{k}{\text{argmin}} \; \text{max} \{ ||f(a_i) - c_k||_2^2 - R_k^2, \; 0 \}$.
\textbf{(v)} Finally, with all other updated variables fixed, we compute the partial gradients of Eq. \ref{eq:dual} w.r.t. the neural network parameters $\mathcal{W}$ and update them using backpropagation with ADAM optimizer. It is to be noted that, exact computation of the partial gradient w.r.t. each node takes $O(NK)$ time due to the second term in Eq. \ref{eq:dual}. To avoid this, we sub-sample a fixed number of nodes from the same community to approximate that term. Similarly, we also sub-sample a fixed number of nodes from the neighborhood of any node \cite{hamilton2017inductive} to approximate the fourth term (homophily loss) Eq. \ref{eq:dual}.
We use the standard parameterization of the ADAM optimizer as in \cite{kingma2014adam}. 
Steps (i) to (v) are run sequentially within a loop and the loop is iterated for a fixed number of times or until the loss (in Eq. 10) converges. $\lambda_i$ is considered as the outlier score of the node $v_i \in V$. Please refer Algorithm \ref{alg:DMGD} for the pseudocode of DMGD.

\begin{algorithm} 
  \small
  \caption{\textbf{DMGD} - Deep Multiclass Graph Description}
  \label{alg:DMGD}
\begin{algorithmic}[1]
      
	\Statex \textbf{Input}: The graph $G=(V,E)$, $|V|=N$, Given or generated feature vector $a_i \in \mathbb{R}^D$ for each node $v_i \in V$, $M$: Dimension of the embedding space, $K$: Number of communities in the graph
    \Statex \textbf{Output}: The node embeddings $f(a_i)$, $i \in [N]$ of the graph $G$, Outlier score $\lambda_i$, $i \in [n]$ of the nodes, Community Centers $c_k \in \mathbb{R}^M$ for each community $k \in [K]$, Community assignment variables $\theta_{i1}, \cdots, \theta_{iK}$ for each node $i \in [N]$. 
	\State Run few epochs of the autoencoder (optimizing just Eq. 3 of the main paper) without considering the other terms to initialize the node embeddings
	\State Run k-means++ algorithm (with number of clusters equals to $K$) on these embeddings to get the initial hard community assignments (i.e., values of $\theta_{ik}$).
	\For{$iter \in \{1,2,\cdots,T\}$}
      \State Use standard quadratic solver CVXOPT qpsolver to update the values of $\Lambda$ (refer to Eq. 10)
      \State Calculate the center of each community $k \in [K]$ by $c_k = \sum\limits_{i=1}^N \lambda_i \theta_{ik} f(a_i)$ from Eq. 8 of the main paper.
      \For{$k \in [K]$}
        \State Calculate the radius $R_k$ as $\underset{i}{\text{min}}\{||f(a_i)-c_k||_2^2 \; | \; \theta_{ik}=1 \; , \; 0 < \lambda_i < \alpha \}$
      \EndFor
      \For{$i \in [N]$}
        \State Reassign the community of the node $v_i$ with the updated embeddings and centers by updating $\theta_{ik} = \underset{k}{\text{argmin}} \; \text{max} \{ ||f(a_i) - c_k||_2^2 - R_k^2, \; 0 \}$.
      \EndFor
      \State Retrain the autoencoder (Section 3 of the main paper) by computing the partial gradients of Eq. 10 of the main paper w.r.t. the neural network parameters $\mathcal{W}$ and update them using backpropagation with ADAM optimizer
    \EndFor
	\end{algorithmic}
  \end{algorithm} 

\textbf{Time and Space Complexity Analysis}:
Initial cluster assignments using k-means++ algorithm takes $O(NMK)$ time, assuming number of iterations needed is a constant, where $M$ is the embedding dimension. CVXOPT has run time of $O(NlogN)$ to updates $\Lambda$. Steps (ii) to (iv) take a total time of $O(NK)$. Followed by that, updates of the neural network again take $O(NK)$ time, thanks to the sub-sampling techniques which also make the time complexity independent of the number of edges in the network, without any significant drop in performance.
Hence the  major bottleneck is the computation of $\Lambda$ using the quadratic solver. To overcome this, one can use the bounds on the number of points having non-zero $\lambda_i$ from Lemmas \ref{lemma:upbound} and \ref{lemma:lowbound}. An improved solution based on this can be addressed in the future.
The space complexity of DMGD is $O(|V| + |E|)$. The graph can be stored as an adjacency list on the disk. DMGD does not need the whole $O(N^2)$ expensive adjacency matrix to work with.


\section{\uppercase{Experimental Evaluation}}\label{sec:exp}
\subsection{Datasets Used and Seeding Outlier}\label{sec:dataset}
One primary goal of DMGD is to handle and detect outliers while generating the embeddings. To the best of our knowledge, there is no publicly available standard network dataset with ground truth outliers. Also for many of the network datasets, underlying community structure does not match the label of the nodes, i.e., two nodes having same ground truth label may not belong to the same community and vice versa. So to validate our algorithm, we use a combination of synthetic (LFR Benchmark: \url{https://bit.ly/2Xx4EJh}) and real world network datasets, described in Table \ref{tab:data}. We also seed 5\% outliers into each dataset by perturbing nodes as follows. To perturb each outlier, we select a node randomly from the dataset. We find the top 20\% nodes from the dataset which are at the farthest distance from the selected node. Finally we randomly sample an equal number of nodes as the degree of the node in the network from neighbors of those 20\% nodes and the neighbors of the selected nodes. Most of the neighbors of the selected node belong to the same community, and most of the farthest nodes belong to different communities. Thus our perturbed outliers have edges to nodes from multiple communities, and satisfy the conditions of a community outlier. Please note, labels of outliers have \textbf{not} been considered for calculating the accuracy of any downstream task. 

\begin{table}
\caption{Summary of the datasets, after planting outliers.}
	\centering
    \resizebox{0.7\columnwidth}{!}{%
	\begin{tabular}{*4c}
	\toprule
	\sffamily{Dataset} & \#Nodes & \#Edges & \#Labels \\
    \hline
	\midrule
	\sffamily{LFR Benchmark}  & 1200 & 5277 & 10 \\
	\sffamily{WebKB} & 877 & 2897  & 5 \\
	\sffamily{Polblogs} & 1224 & 19025  & 2 \\
	\sffamily{Cora} & 2708 & 5429  & 7 \\
	\sffamily{Pubmed} & 19717 & 44338  & 3 \\
\bottomrule
	\end{tabular}
    }
	\label{tab:data}
\end{table}

\subsection{Baseline Algorithms and Experimental Setup}
The proposed algorithm DMGD is unsupervised in nature. So as baselines, we choose only well-known unsupervised embedding algorithms which can work even without attributes of the nodes. The baselines are DeepWalk \cite{perozzi2014deepwalk}, node2vec \cite{grover2016node2vec}, LINE \cite{tang2015line}, SDNE \cite{wang2016structural}, GraphSAGE \cite{hamilton2017inductive} (unsupervised version) and DGI \cite{velivckovic2018deep} for all the downstream tasks. We use the publicly available implementation of these algorithms, with default hyper parameter settings.
We do not consider node attributes as the focus is to exploit only the network structure.
Additionally, we consider SVDD and Deep SVDD only for the experiments on outlier recall, as these two algorithms are not meant for network embedding. We generate node embeddings using Eq. \ref{eq:loss1} before applying SVDD. We have also used following two purely graph based algorithms (not for node embedding): AltQP-Inc \cite{tong2011non} for outlier detection and SBMF \cite{zhang2013overlapping} for community detection, and included the results for the respective tasks.

Embedding dimension is fixed at 16 for all the algorithms, on all the datasets, except Pubmed. For Pubmed, the embedding dimension is 32 as it is larger in size. We tried with increased embedding dimensions also (for e.g., 128), but there is no significant improvement in results. Encoder and decoder in DMGD contain two layers each for all the datasets. We use leaky-ReLU activation for non-linearity in all the layers except the last one, which has ReLU activation. We train autoencoder using ADAM \cite{kingma2014adam} optimizer with default parameters. 

\begin{figure*}[h!]
   \begin{subfigure}[b]{\linewidth}
    \centering \includegraphics[width=0.7\linewidth]{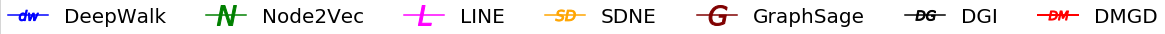}
    \vspace{0.3cm}
  \end{subfigure}
  \centering
  \begin{subfigure}[b]{0.19\linewidth}
    \includegraphics[width=\linewidth]{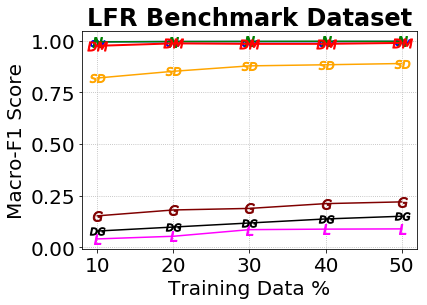}
  \end{subfigure}
  \begin{subfigure}[b]{0.19\linewidth}
    \includegraphics[width=\linewidth]{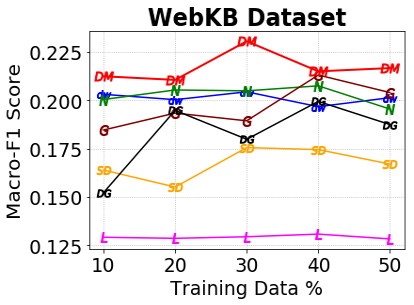}
  \end{subfigure}
  \begin{subfigure}[b]{0.19\linewidth}
    \includegraphics[width=\linewidth]{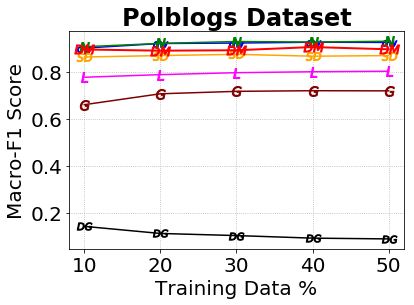}
  \end{subfigure}
  \begin{subfigure}[b]{0.19\linewidth}
    \includegraphics[width=\linewidth]{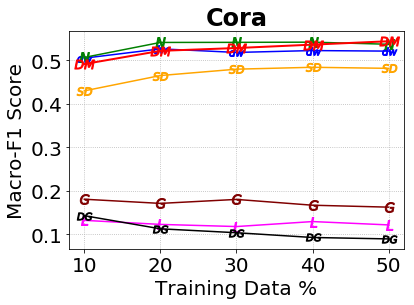}
  \end{subfigure}
  \begin{subfigure}[b]{0.19\linewidth}
    \includegraphics[width=\linewidth]{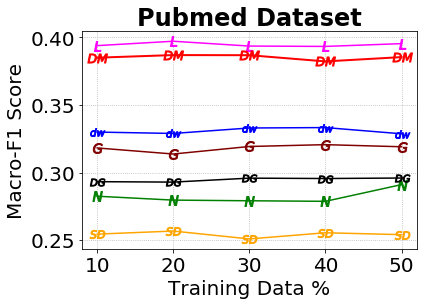}
  \end{subfigure}

   \caption{Accuracy of Node classification with Logistic Regression}
	\label{fig:classi}
\end{figure*}

\subsection{Setting the hyper parameters of DMGD}\label{sec:hyperparameters}
Like many other ML algorithms, we also assume to know the number of communities $K$ of a network. We theoretically show the relation between the hyper parameter $\alpha$ and the number of outliers detected by DMGD. If the expected number of outliers is known a priori, $\alpha$ can be set accordingly. The ratio $\frac{\gamma}{\beta}$, after taking $\beta$ as common from the last two components of Eq. 5 of the paper, weights the autoencoder reconstruction loss and the homophily loss and can be set as shown in \cite{wang2016structural}.
$\beta$ balances the outlier and community detection parts of DMGD with the embedding generation part. Increasing value of $\beta$ would give more importance to generating the generic node embeddings, where decreasing it would be to minimize the effect of outliers and give importance on the community structure of the network. To set the value of $\beta$, we use the standard grid search which minimizes both total loss and the individual component losses of the DMGD cost function.

\begin{figure}
\centering
\begin{subfigure}[b]{\linewidth}
    \centering \includegraphics[width=0.8\linewidth]{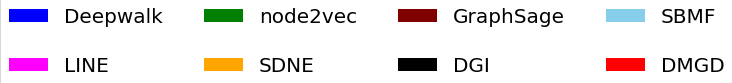}
  \end{subfigure}
\centering
\begin{subfigure}[t]{\linewidth}
  \centering
  \includegraphics[width=0.8\linewidth]{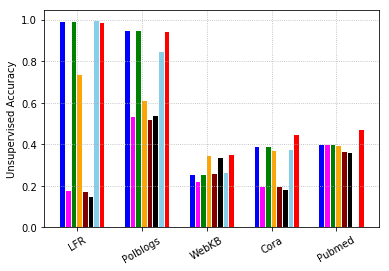}
\end{subfigure}
\caption{Unsupervised accuracy for community detection}
\label{fig:clustering}
\end{figure}

\subsection{Outlier Detection, Community Detection and Node Classification}\label{sec:tasks}
Outlier detection is extremely important for network embedding, as discussed in Section \ref{sec:intro}. We use $\lambda_i$ as the outlier score of the node $v_i \in V$ for DMGD (Sec. \ref{sec:opt}). SVDD and Deep SVDD also produce outlier scores of the nodes directly. For other baselines, we use isolation forest algorithm \cite{liu2008isolation} on the generated embeddings to get outlier scores of the nodes. Each seeded dataset has 5\% outliers. So we plot the outlier recall from the top 5\% to 25\% of the nodes in the ranked list (L) with respect to the seeded outliers in Fig. \ref{fig:outlierRecall}. Clearly, DMGD outperforms all the baseline algorithms for outlier recall. As it adheres multiple community structure of the network, it is able to detect outliers which lie between the communities. Deep SVDD, due to its high non-linear nature, turns out to be more consistent than other baselines. Most of the standard graph embedding algorithms like node2vec and GraphSAGE suffer as they do not process outliers while generating the embeddings.

DMGD also outputs the community assignment of the nodes in the graph though the set of variables $\Theta$ (Sec. \ref{sec:opt}). Here we check the quality of the communities produced by DMGD, with respect to the ground truth labeling of the datasets. For the baseline algorithms (except for SBMF), we give the node embeddings as input to KMeans++ \cite{kmeans++}. To judge the quality of clustering, we use unsupervised clustering accuracy \cite{xie2016unsupervised}.
Figure \ref{fig:clustering} shows that DMGD performs better or almost as good as the best of the baselines for community detection. As DMGD integrates community detection with graph embedding and outlier detection, the output communities are more optimal than most of the baselines which finds communities by post-processing the embeddings. We could not present the result of SBMF on Pubmed as the runtime exceeds more than 3 days.

To compare the quality of node embeddings generated by DMGD to that of the baselines, we also consider node classification. We vary the training size from 10\% to 50\%. We train a logistic regression classifier on the training set of embeddings (along with the class labels) and check the performance on the test set by using Macro F1 score. Figure \ref{fig:classi} shows the performance of node classification for all the publicly available datasets we used. 
Though the optimization of DMGD explicitly handles communities and outliers in node embeddings, its performance for node classification is highly competitive to state-of-the-art network embedding algorithms. On WebKB, DMGD turns out to be the best performer, while for other datasets, it is always very close to the best of the baselines. 
DeepWalk, node2vec, SDNE and LINE also perform good for node classification task depending on the datasets.

\begin{table}
\caption{Classification (macro \& micro F1 with training size 30\%) and clustering (unsupervised accuracy) performance on the unseeded and seeded versions of Cora.}
\label{tab:out_eff_cora}
\centering
\resizebox{1.0\columnwidth}{!}{%
\begin{tabular}{c|c|ccccc}
\toprule
\multicolumn{2}{c}{}                   & node2vec & LINE & GraphSAGE & DGI & DMGD \\
\midrule
\multirow{2}{*}{Macro-F1 (\%)}   & unseeded & \textbf{55.99} & 16.02 & 20.55 & 10.84 & 53.05    \\
                            & seeded     & \textbf{54.11}  & 11.80  & 18.03 & 10.38 & 52.80     \\
\cline{1-7}
\multirow{2}{*}{Micro-F1 (\%)}   & unseeded & \textbf{60.72}	& 35.32	& 34.86	& 31.54 & 56.32     \\
                            & seeded     & 55.99 & 29.16 & 31.65 & 27.85 &  \textbf{57.24} \\
\cline{1-7}
\multirow{2}{*}{Clustering (\%)} & unseeded   & 39.51	& 27.51	& 22.34	& 20.08 & \textbf{43.83} \\
                            & seeded & 39.31	& 19.77 & 20.12 & 18.45  & \textbf{44.86} \\
\bottomrule
\end{tabular}
}
\end{table}


\subsection{Influence of Outliers on Node Embeddings}\label{sec:effect_out}
Here, we empirically show the negative influence of outliers on the node embeddings and thus motivate the problem again.
To show the effect of outliers, we run DMGD and some of the better performing and diverse baseline algorithms (due to limitation of page) on both the unseeded (original) and seeded (with outliers) versions of Cora dataset in Table \ref{tab:out_eff_cora}. 
We use node classification and clustering as the downstream tasks to show the effect. Clearly, most of the baseline algorithms are affected because of the presence of the community outliers. The performance of an algorithm is generally better in the unseeded version of the dataset than the seeded one. This can be seen by comparing two consecutive rows for a metric in Table \ref{tab:out_eff_cora}.
For DMGD, the adverse effect is less as it is resistant to outliers. In some of the cases, there is some marginal improvement in the performance of DMGD on the seeded dataset.

\section{\uppercase{Discussion and Future Work}}
In this work, we proposed an unsupervised algorithm DMGD, which integrates node embedding, minimizing the effect of outliers and community detection into a single optimization framework. 
We also show the theoretical bounds on the number of outliers detected by it. One shortcoming of DMGD is that, it depends heavily on the community structure of the network, which may not be so prominent in some real network. In future, we would like to address this issue. We would also like to conduct experiments on networks with overlapping communities to check the performance of DMGD in that case.


\bibliography{ecai}
\end{document}